\pdfpagewidth=\paperwidth
\pdfpageheight=\paperheight
\documentclass[12pt, letterpaper]{article}
\pdfoutput=1

\DeclareFixedFont{\MyTitleFont}{OT1}{ptm}{m}{n}{18pt}
\DeclareFixedFont{\MyAuthorFont}{OT1}{ptm}{m}{n}{13pt}
\DeclareFixedFont{\MyAbstractTitleFont}{OT1}{ptm}{m}{it}{12pt}
\DeclareFixedFont{\MyAbstractFont}{OT1}{ptm}{m}{it}{11pt}
\DeclareFixedFont{\MySubtitleFont}{OT1}{ptm}{m}{n}{15pt}
\DeclareFixedFont{\MySubSubtitleFont}{OT1}{ptm}{m}{n}{13pt}
\DeclareFixedFont{\MySubSubSubtitleFont}{OT1}{ptm}{m}{n}{11pt}
\DeclareFixedFont{\MyTextFont}{OT1}{ptm}{m}{n}{11pt}

\usepackage{geometry}
\usepackage{titling}

\title{\MyTitleFont A Simplified Approach to Analyze Complementary Sensitivity Trade-offs in Continuous-Time and Discrete-Time Systems   \vspace{0em}}

\author{\MyAuthorFont Neng Wan$^1$, Dapeng Li$^2$, and Naira Hovakimyan$^1$}
\date{}

\usepackage{titlesec}

\usepackage[marginal]{footmisc}
\usepackage{abstract}

\usepackage[sort]{cite}

\usepackage{indentfirst}
\usepackage{amsthm}
\usepackage{amsfonts}
\usepackage{amsmath}
\usepackage{booktabs}
\usepackage{bm}
\usepackage{mathptmx}
\usepackage{subfigure}
\usepackage{enumerate}
\usepackage[mathcal]{euscript}
\usepackage{float}
\usepackage{graphics} 
\usepackage{epsfig} 
\usepackage[bookmarks]{hyperref}
\hypersetup{colorlinks = true, citecolor = red, linkcolor = blue, urlcolor = black}

\newtheorem{theorem}{Theorem}
\newtheorem{remark}{Remark}

\newtheorem{lemma}{Lemma}

\newtheorem{corollary}[theorem]{Corollary}

\allowdisplaybreaks

\begin{document}

\maketitle


\footnotetext[0]{*This work was supported in part by AFOSR and NSF. \vspace{0.3em}}
\footnotetext[0]{$^{1}$Neng Wan and Naira Hovakimyan are with the Department of Mechanical Science and Engineering, University of Illinois at Urbana-Champaign, Urbana, IL 61801, USA. {\tt\small \{nengwan2, nhovakim\}@illinois.edu}.}
\footnotetext[0]{$^{2}$Dapeng Li is a Principal Scientist with the JD.com Silicon Valley Research Center, Santa Clara, CA 95054, USA. {\tt\small dapeng.li@jd.com}.}
\vspace{-4em}

\begin{abstract}
{\vspace{0.5em}
	A simplified approach is proposed to investigate the continuous-time and discrete-time complementary sensitivity Bode integrals (CSBIs) in this note. For continuous-time feedback systems with unbounded frequency domain, the CSBI weighted by $1/\omega^2$ is considered, where this simplified method reveals a more explicit relationship between the value of CSBI and the structure of the open-loop transfer function. With a minor modification of this method, the CSBI of discrete-time system is derived, and illustrative examples are provided. Compared with the existing results on CSBI, neither Cauchy integral theorem nor Poisson integral formula are used throughout the analysis, and the analytic constraint on the integrand is removed.
}
\end{abstract}
\vspace{-0.5em}

\titleformat*{\section}{\centering\MySubtitleFont}
\titlespacing*{\section}{0em}{1.25em}{1.25em}[0em]

\section{Introduction}\label{sec1}

This technical note extends the {\em simplified approach} for analysis of sensitivity Bode integrals from~\cite{Wu_TAC_1992} to complementary sensitivity Bode integrals (CSBIs). Sensitivity function  and complementary sensitivity function  are two critical transfer functions that provide insights into the influence of external disturbance on the error signal and the measurement output, respectively. Freudenberg and Looze showed in~\cite{Freudenberg_1985, Freudenberg_1987} that the integral over all frequencies of the logarithm of the absolute value of sensitivity function, $\ln |S(s)|$, is proportional to the sum of the unstable open-loop poles. Meanwhile, motivated by the well-known result for sensitivity function and complementary sensitivity functions~\cite{Seron_2012}, $S(s) + T(s) = 1$, it is natural  to believe that similar trade-off should also exist for $\ln |T(s)|$. However, as $s\rightarrow \infty$, the integrand $\ln |T(s)|$ and the corresponding integral grow to infinity in continuous-time systems, which puzzled the researchers for a few years~\cite{Sung_IJC_1989}.

Several efforts have been made to tackle this issue on complementary sensitivity Bode integral (CSBI) for continuous-time systems. One of the earliest results of Freudenberg and Looze~\cite{Freudenberg_1985} exploited the harmonic property of $\ln |T(s)|$ to define CSBI by multiplying this function with a Poisson kernel and using Poisson integral formula with a limiting argument~\cite{Seron_2012}. A more concise work on continuous-time CSBI was later done by Middleton~\cite{Middleton_1991}, who weighted $\ln |T(s)|$ by $1/\omega^2$ and adopted the Cauchy integral theorem. The  frequency inversion and Cauchy integral theorem in~\cite{Middleton_1991} require the inverse frequency function $\ln |T(1/s)|$ be analytic at $s = \infty$, such that it can be expanded as a Laurent series~\cite{Seron_2012}. In~\cite{Yu_IJRNC_2015} Yu \textit{et al.}  studied a CSBI weighted by $1 / (s^2 + \alpha^2)^k$, where $\alpha \in \mathbb{R}$ and $k \geq 1$, and an information-theoretic approach to derive the CSBI for continuous-time stochastic systems was presented in~\cite{Wan_CDC_2018}. The results on discrete-time CSBI, which has a bounded frequency domain, can be found in~\cite{Sung_IJC_1989, Okano_Auto_2009, Ishii_SCL_2011}.

The simplified approach from~\cite{Wu_TAC_1992} is extended here to analyze the CSBIs for both continuous-time and discrete-time systems. Compared with the prevailing results on CSBI of deterministic systems~\cite{Freudenberg_1985, Sung_IJC_1989, Middleton_1991, Seron_2012, Yu_IJRNC_2015}, the salient feature of this method is that neither Cauchy integral theorem nor Poisson integral formula are invoked when deriving the CSBIs, which consequently allows  to remove the analytic (harmonic) constraints on the integrands of CSBIs. In addition to a new approach to derive CSBI, this simplified approach also provides a more explicit explanation on how the complementary sensitivity property is impacted by the structure of an open-loop transfer function, \textit{e.g.} the distributions of zeros and poles, relative degree, number of pure integrators and leading coefficient. For continuous-time systems, we study the CSBI weighted by $1/\omega^2$ similar to~\cite{Middleton_1991}, and with a slight modification, the simplified approach is applied to investigating the discrete-time CSBI. A few illustrative examples are given at the end of this note.

This note is organized as follows: \hyperref[sec2]{Section 2} introduces the preliminaries; \hyperref[sec3]{Section 3} studies the CSBI of continuous-time systems; \hyperref[sec4]{Section 4} investigates the CSBI of discrete-time systems; several illustrative examples are shown in~\hyperref[sec5]{Section 5}, and~\hyperref[sec6]{Section 6} draws the conclusions.

\noindent \textit{Notation}: In this paper, we use $\ln(\cdot)$ to denote natural logarithm with the base of the mathematical constant $\textrm{e}$ and $\log(\cdot)$ to denote the logarithm with base $2$. For a complex number $a$, $|a|$ stands for the modulus. Complex variables are denoted as $s = j\omega$ and $z = \textrm{e}^{j\omega}$.

\section{Preliminaries and Problem Formulation}\label{sec2}

Background knowledge and some preliminary results on CSBI are stated in this section. Consider the following block diagram of a general feedback system, which can be used to describe both continuous-time and discrete-time systems.

\begin{figure}[H]
	\centering
	\vspace{-1em}\includegraphics[width=0.5\textwidth]{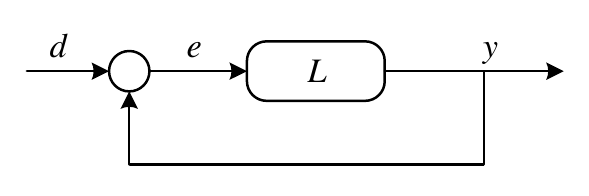}\vspace{-1em}
	\caption{\small General feedback system.}
\end{figure}

\noindent Here $L$ denotes the open-loop transfer function, $d$ is the external disturbance, $e$ is the error signal, and $y$ is the measurement output. The complementary sensitivity function $T(s)$ (or $T(z)$) is defined as the transfer function from external disturbance $d$ to measurement output $y$.

\titleformat*{\subsection}{\MySubSubtitleFont}
\titlespacing*{\subsection}{0em}{0.75em}{0.75em}[0em]
\subsection{Continuous-Time System}
For a continuous-time plant model $G(s)$ and a control mapping $C(s)$ the open-loop transfer function $L(s)$ can be generally expressed as
\begin{equation}\label{cont_tf}
	L(s) = G(s) \cdot C(s) = K \cdot \dfrac{\prod_{i=1}^{m}(s - z_i)}{s^{n-l} \cdot \prod_{i=1}^{l}(s - p_i)},
\end{equation}
where $K \in \mathbb{R}$, the relative degree is $\nu = n - m$ with $m \leq n$, $n-l \geq 0$ denotes the number of pure integrators, $z_i$ and $p_i$ respectively denote the zeros and poles of $L(s)$, and no $z_i$ or $p_i$ is at $s = 0$. When $\nu = 1$, the leading coefficient $K = \lim_{s \rightarrow \infty} s L(s)$. The complementary sensitivity function $T(s)$ for this continuous-time system is defined as
\begin{equation}\label{complementary}
	T(s) = \dfrac{L(s)}{1 + L(s)}.
\end{equation}
In this note, we consider the continuous-time CSBI defined as follows:
\begin{equation}\label{Cont_Bode}
	\dfrac{1}{2\pi} \int_{-\infty}^{\infty} \ln |T(s)| \ \dfrac{d\omega}{\omega^2},
\end{equation}
where a weighting function $1 / \omega^2$ is involved~\cite{Middleton_1991, Seron_2012}. Consider the following frequency transformation
\begin{equation}\label{freq_trans}
s = j\omega = \dfrac{1}{j\tilde{\omega}} = \tilde{s}^{-1},
\end{equation}
where the frequency variables satisfy $\omega = - \left(\tilde{\omega}\right)^{-1}$. By change of variables, we can rewrite CSBI in~\eqref{Cont_Bode} as follows
\begin{equation*}
	\dfrac{1}{2\pi} \int_{-\infty}^{\infty} \ln |\tilde{T}\left(\tilde{s}\right)| d\tilde{\omega},
\end{equation*}
where $\tilde{T}(\tilde{s}) = T(s)$. The following lemma states the earlier result on CSBI, which was obtained by resorting to Cauchy integral theorem~\cite{Middleton_1991, Seron_2012}.
\begin{lemma}\label{lem1}
	Let $z_{u_i}$'s be the non-minimum phase zeros of open-loop transfer function $L(s)$, and suppose that $L(0) \neq 0$. Then, assuming closed-loop stability, if $L(s)$ is a proper rational function, then
	\begin{equation*}
		\dfrac{1}{2\pi}\int_{-\infty}^{\infty} \ln \left| \dfrac{T(s)}{T(0)} \right| \dfrac{d\omega}{\omega^2} =  \sum_{i}z_{u_i}^{-1} + \dfrac{1}{2 \cdot T(0)}\lim_{s \rightarrow 0} \dfrac{d T(s)}{ds}.
	\end{equation*}
\end{lemma}
\noindent The theoretical basis of the simplified approach for continuous-time system is stated in the following lemma~\cite{Wu_TAC_1992}.
\begin{lemma}\label{lem2}
	For complex numbers $a$ and $b$, we have
	\begin{equation}\label{lem2_eq}
		\int_{-\infty}^{\infty} \ln \left| \dfrac{(j\omega - a)}{(j\omega - b)} \right|^2 d\omega = 2\pi \cdot \left( \left| \mathrm{Re} \ a \right| - \left| \mathrm{Re} \ b \right| \right).
	\end{equation}
\end{lemma}
\begin{remark}
	The proof of~\hyperref[lem2]{Lemma~2} only requires some elementary techniques, such as integration by parts. Instead of weighting the left-hand side (LHS) of~\eqref{lem2_eq} by $1 / \omega^2$ and deriving another identity, we find that employing~\hyperref[lem2]{Lemma~2} to derive continuous-time CSBI can give  more insights into the interactions between the value of CSBI and the structure of the open-loop transfer function $L(s)$.
\end{remark}

\subsection{Discrete-Time System}
With a discrete-time plant model $G(z)$ and control mapping $C(z)$, the open-loop transfer function $L(z)$ can generally be expressed as
\begin{equation}\label{dis_tf}
	L(z) = G(z) \cdot C(z) = K \cdot \dfrac{\prod_{i=1}^{m}(z - z_i)}{\prod_{i=1}^{n}(z - p_i)},
\end{equation}
where $K \in \mathbb{R}$, relative degree is $\nu = n - m \geq 0$, $z_i$ and $p_i$ are respectively the zeros and poles, and $z_i \neq 0$. Compared with the open-loop transfer function for continuous-time system~\eqref{cont_tf}, since frequency transformation is not involved when deriving the discrete-time CSBI, unit delays are not explicitly expressed in~\eqref{dis_tf}, and we allow $p_i = 0$ in discrete-time system. The discrete-time complementary sensitivity function $T(z)$ is then defined as
\begin{equation*}
	T(z) = \dfrac{L(z)}{1 + L(z)}.
\end{equation*}
Since the frequency domain of discrete-time system is bounded, $\omega \in [-\pi, \pi]$, we consider the following type of CSBI without weighting function
\begin{equation}\label{eq7}
	\dfrac{1}{2\pi}\int_{-\pi}^{\pi} \log |T(z)| d\omega.
\end{equation}
Previous result on the discrete-time CSBI is claimed in the following lemma, which was also derived on the basis of Cauchy integral theorem~\cite{Sung_IJC_1989}.
\begin{lemma}\label{lem3}
	Let $z_{u_i}$'s be the strictly unstable zeros of open-loop transfer function $L(z)$. Then, assuming closed-loop stability, if $L(z)$ is a proper rational function, we have
	\begin{equation*}
		\dfrac{1}{2\pi}\int_{-\pi}^{\pi} \log \left|T(z)\right| d\omega =  \sum_{i} \log |z_{u_i}| + \log |K|,
	\end{equation*}
	where $z_{u_i}$ denotes the unstable zeros in $L(z)$, and $K$ is the leading coefficient of the numerator of $L(z)$, when the denominator is monic.
\end{lemma}
\noindent The fundamental tool for analyzing discrete-time CSBI is stated in the following lemma, whose proof only requires elementary techniques and is available in~\cite{Wu_TAC_1992}.
\begin{lemma}\label{lem4}
	For a complex number $a$, we have
	\begin{equation*}
		\int_{-\pi}^{\pi} \log \left| \mathrm{e}^{j\omega} - a\right|^2 d\omega =
		\begin{cases}
			0, \qquad \qquad \quad & \textrm{if }|a| \leq 1;\\
			2\pi \cdot \log |a|^2, & \textrm{if }|a| > 1.
		\end{cases}
	\end{equation*}
\end{lemma}

\section{Continuous-Time Complementary Sensitivity Bode Integral}\label{sec3}

We investigate the continuous-time CSBI in this section. The results are stated in two categories, namely when relative degree $\nu \geq 1$ and $\nu = 0$. Under each category, we show how the value of CSBI is related to the amount of pure integrators, as well as the leading coefficient of the open-loop transfer function. First, we consider a more general scenario when the open-loop transfer function is strictly proper, \textit{i.e.} $\nu \geq 1$.

\begin{theorem}\label{thm1}
	For an open-loop transfer function $L(s)$ with relative degree $\nu \geq 1$ and stable closed-loop system, the continuous-time CSBI satisfies
	\begin{equation}\label{thm1_eq}
			\dfrac{1}{2\pi}\int_{-\infty}^{\infty} \log |T(s)| \dfrac{d\omega}{\omega^2} =
			\begin{cases}
			\sum_{i} \mathrm{Re} \ z^{-1}_{u_i},   & \mathrm{if} \ 0 \leq l \leq n - 2; \\
			 \sum_{i} \mathrm{Re} \ z_{u_i}^{-1} - \dfrac{1}{2K} \cdot \dfrac{\prod_{i=1}^{n-1}(-p_i)}{\prod_{i=1}^{m}(-z_i)}, \hspace{16pt} & \mathrm{if} \ l=n-1; \\
			 \pm \infty, & \mathrm{otherwise},
			\end{cases}
	\end{equation}
	where $p_i$, $z_i$ and $z_{u_i}$ respectively denote the poles, zeros, and non-minimum phase zeros in $L(s)$, and $n-l$ is the amount of pure integrators in $L(s)$.
\end{theorem}

\begin{proof}[\textbf{Proof.}]
	When relative degree $\nu \geq 1$, the open-loop transfer function defined in~\eqref{cont_tf} can be expressed as follows
	\begin{equation}\label{cont_tf2}
		L(s) = K \cdot \dfrac{\prod_{i=1}^{m}(s - z_i)}{s^{n-l} \cdot \prod_{i=1}^{l}(s - p_i)},
	\end{equation}
	where $l \leq n$, $m + 1 \leq n$, $K \in \mathbb{R}$, and no $z_i$ or $p_i$ is at $s = 0$. Substitute~\eqref{cont_tf2} into~\eqref{complementary}, and rewrite the complementary sensitivity function $T(s)$ in the following two equivalent forms
	\begin{align}
		T_1(s) & = T(s) = \dfrac{ K \cdot \prod_{i=1}^{m}(s - z_i) }{s^{n-l} \cdot \prod_{i=1}^{l}(s - p_i) + {K \cdot \prod_{i=1}^{m}(s - z_i) }}, \label{C_T1}\\	
		T_2(s) & = T(s) = K \cdot \dfrac{\prod_{i=1}^{m}(s - z_i) }{ \prod_{i=1}^n (s - r_i) }, \label{C_T2}
	\end{align}
	where $r_i$'s denote the closed-loop poles with negative real parts. Applying frequency transformation~\eqref{freq_trans} to complementary sensitivity functions~\eqref{C_T1} and~\eqref{C_T2} gives
	\begin{align}
		\tilde{T}_1(\tilde{s}) & = T_1(s) = \dfrac{ K \cdot \prod_{i=1}^{m}(\tilde{s}^{-1} - z_i) }{\tilde{s}^{l-n} \cdot \prod_{i=1}^{l}(\tilde{s}^{-1} - p_i)+ {K \cdot \prod_{i=1}^{m}(\tilde{s}^{-1} - z_i) }}, \label{C1_T1}\\
		\tilde{T}_2(\tilde{s}) & = T_2(s) = K \cdot  \dfrac{ \prod_{i=1}^{m}(\tilde{s}^{-1} - z_i) }{ \prod_{i=1}^n (\tilde{s}^{-1} - r_i) }. \label{C1_T2}
	\end{align}
	Multiplying the numerators and denominators of~\eqref{C1_T1} and~\eqref{C1_T2} by $\tilde{s}^n$ and with some algebraic manipulations, we have
	\begin{align}
		\tilde{T}_1(\tilde{s}) & =   \dfrac{K \cdot \tilde{s}^{n-m} \prod_{i=1}^{m}(1 - z_i \cdot \tilde{s})  }{\prod_{i=1}^{l}(1 - p_i \cdot \tilde{s}) + {K \cdot \tilde{s}^{n-m}  \prod_{i=1}^{m}(1 - z_i \cdot \tilde{s})}} \label{C2_T1}\\
		& =  \dfrac{ K \prod_{i=1}^{m}(-z_i) \cdot \tilde{s}^{n-m} \prod_{i=1}^{m}(\tilde{s} - z_i^{-1})  }{ \prod_{i=1}^l\left(-p_i\right) \cdot  \prod_{i=1}^{l}(\tilde{s} - p_i^{-1}) + {K \prod_{i=1}^{m} \left(-z_i\right) \cdot \tilde{s}^{n-m} \prod_{i=1}^{m}(\tilde{s} - z_i^{-1})  }}, \nonumber \\
		\tilde{T}_2(\tilde{s}) & = K \cdot \dfrac{ \tilde{s}^{n-m} \prod_{i=1}^{m}(1 - z_i \cdot \tilde{s})  }{ \prod_{i=1}^n (1 - r_i \cdot \tilde{s}) } = K \cdot \dfrac{\prod_{i=1}^{m}\left(-z_i\right)}{\prod_{i=1}^{n}\left(-r_i\right)} \cdot \dfrac{\tilde{s}^{n-m} \prod_{i=1}^{m}(\tilde{s} - z_i^{-1})}{\prod_{i=1}^n (\tilde{s} - r_i^{-1})}. \label{C2_T2}
	\end{align}
	 Some relationship among $p_i$, $z_i$, and $r_i$ can be implied from~\eqref{C_T1}-\eqref{C2_T2}. Equating the denominators of~\eqref{C_T1} and~\eqref{C_T2}, we have
	 \begin{equation}\label{expand1}
	 	s^{n-l} \cdot \prod_{i=1}^{l}(s - p_i) + {K \cdot \prod_{i=1}^{m}(s - z_i) } = \prod_{i=1}^n (s - r_i).
	 \end{equation}
	 Expanding both sides of~\eqref{expand1} yields
	\begin{equation}\label{C_Den}
		\begin{split}
			s^n - \left( \sum_{i=1}^{l}p_i \right) s^{n-1} + \cdots + \prod_{i=1}^{l}\left(-p_i\right)  s^{n-l} & + K \cdot \left[s^m -\left( \sum_{i=1}^{m} z_i \right) \cdot s^{m-1} + \cdots + \prod_{i=1}^{m}\left(-z_i\right)\right] \\
			&= s^n - \left(\sum_{i=1}^{n} r_i\right)s^{n - 1} + \cdots + \prod_{i=1}^{n}\left(-r_i\right).
		\end{split}
	\end{equation}
	Since complex roots always come in conjugate pairs, the products $\prod_{i=1}^{l}(-p_i)$, $\prod_{i=1}^{m}(-z_i)$ and $\prod_{i=1}^{n}\left(-r_i\right)$ in~\eqref{C_Den} are all real. Equating the denominators in~\eqref{C2_T1} and~\eqref{C2_T2} gives
	\begin{equation}\label{expand2}
		\prod_{i=1}^{l}(-p_i) \prod_{i=1}^{l}(\tilde{s} - p_i^{-1}) + K \prod_{i=1}^{m}(-z_i) \cdot \tilde{s}^{n-m}  \prod_{i=1}^{m}(\tilde{s} - z_i^{-1}) = \prod_{i=1}^{n}(-r_i) \cdot \prod_{i=1}^{n}(\tilde{s} - r_i^{-1}).
	\end{equation}
	Expanding both sides of~\eqref{expand2} yields
	\begin{equation}\label{C_Den2}
		\begin{split}
			\prod_{i=1}^{l}(-p_i) \left[ \tilde{s}^l - \left( \sum_{i=1}^{l} p_i^{-1} \right) \tilde{s}^{l-1} + \cdots + \prod_{i=1}^{l}\left( -p_i^{-1} \right) \right] + K \prod_{i=1}^{m}(-z_i) \Bigg[ \tilde{s}^{n} - \left( \sum_{i=1}^{m} z_i^{-1} \right)\tilde{s}^{n-1} + \cdots\\
			+ \prod_{i=1}^{m}\left(-z_i^{-1}\right)\tilde{s}^{n-m} \Bigg] = \prod_{i=1}^{n}\left(-r_i\right) \left[ \tilde{s}^n - \left( \sum_{i=1}^{n} r_i^{-1} \right) \tilde{s}^{n-1} + \cdots + \prod_{i=1}^{n}\left(-r_i^{-1}\right)  \right].
		\end{split}
	\end{equation}
	The value of CSBI varies depending on the amount of pure integrators in the open-loop transfer function $L(s)$.
	
	\vspace{0.5em}	
	
	\noindent \textit{Case 1}: No pure integrator exists in $L(s)$, \textit{i.e.} $l = n$.
		
		When $l = n$, equating the constant terms in~\eqref{C_Den} gives the following identity
		\begin{equation}\label{coef1_case1}
			\prod_{i=1}^{n}\left(-p_i\right) + K  \prod_{i=1}^{m}  \left(-z_i\right) = \prod_{i=1}^{n}\left(-r_i\right).
		\end{equation}
		Equating the coefficients of terms $\tilde{s}^{n-1}$ in~\eqref{C_Den2} yields
		\begin{equation}\label{coef2_case1}
			\prod_{i=1}^{n}(-p_i) \left(\sum_{i=1}^{n} p_i^{-1} \right)  + K \prod_{i=1}^{m}(-z_i)  \left( \sum_{i=1}^{m}z_i^{-1} \right)  = \prod_{i=1}^{n}\left(-r_i\right) \left( \sum_{i=1}^{n} r_i^{-1} \right).
		\end{equation}
		Applying~\hyperref[lem2]{Lemma~2} to $\tilde{T}_2(\tilde{s})$ in~\eqref{C2_T2} and substituting~\eqref{coef1_case1} into the result, the CSBI satisfies
	\begin{equation}\label{Cont_Bode2}
		\begin{split}
			&\dfrac{1}{2\pi}\int_{-\infty}^{\infty} |T(s)| \ \dfrac{d\omega}{\omega^2}  = 	\dfrac{1}{2\pi} \int_{-\infty}^{\infty} \ln |\tilde{T}_2\left(\tilde{s}\right)| d\tilde{\omega}\\
			& \hspace{4em} = \dfrac{1}{2} \cdot \left[ \sum_{i=1}^{m} \left| \textrm{Re} \ z_i^{-1}\right| - \sum_{i=1}^{n} \left| \textrm{Re} \ r_i^{-1} \right| \right]    + \dfrac{1}{\pi} \int_{-\infty}^{\infty} \ln \left|K \cdot \dfrac{\prod_{i=1}^{m}\left(-z_i\right)}{\prod_{i=1}^{n}\left(-r_i\right)}\right| d\tilde{\omega}\\
			& \hspace{4em} = \dfrac{1}{2} \cdot \left[ \sum_{i=1}^{m} \left| \textrm{Re} \ z_i^{-1}\right| - \sum_{i=1}^{n} \left| \textrm{Re} \ r_i^{-1} \right| \right]   + \dfrac{1}{\pi} \int_{-\infty}^{\infty} \ln \dfrac{\left|K \prod_{i=1}^{m}\left(-z_i\right)\right|}{\left|\prod_{i=1}^{n}\left(-p_i\right) + K \prod_{i=1}^{m}  \left(-z_i\right)\right| } d\tilde{\omega}.
		\end{split}
	\end{equation}	
	Since the first two terms on the right-hand side (RHS) of~\eqref{Cont_Bode2} are bounded, the CSBI is bounded if and only if $\left|K \prod_{i=1}^{m}\left(-z_i\right)\right| = \left|\prod_{i=1}^{n}\left(-p_i\right) + K \prod_{i=1}^{m}  \left(-z_i\right)\right|$. This condition can be attained if: i) at least one $p_i = 0$, which contradicts the previous assumption that $p_i \neq 0$, and hence CSBI is undefined for this case; or ii) $\prod_{i=1}^{n}\left(-p_i\right) \allowbreak= -2K \prod_{i=1}^{m}  \left(-z_i\right)$, which was not stated in the previous results and is also omitted in the claim of~\hyperref[thm1]{Theorem~1}, since this condition is trivial and can rarely be satisfied in practice. With~\eqref{coef2_case1} and some further analysis, the CSBI in condition ii) satisfies $(2\pi)^{-1}\int_{-\infty}^{\infty} |T(s)| / \omega^2 \ d\omega = \sum_{i=1}^{n} \mathrm{Re} \ p_i^{-1} - \sum_{i} \mathrm{Re} \ z_{s_i}^{-1}$, where $z_{s_i}$ denotes the minimum phase zeros of $L(s)$. For most open-loop transfer functions without pure integrator, the corresponding CSBIs are unbounded. Specifically, when $| \prod_{i=1}^{n}\left(-p_i\right) + K \cdot \prod_{i=1}^{m}  \left(-z_i\right)  | < | K \cdot \prod_{i=1}^{m}\left(-z_i\right) |$, the CSBI is negative infinity; while when $| \prod_{i=1}^{n}\left(-p_i\right) + K \cdot \prod_{i=1}^{m}  \left(-z_i\right)  |  >  | K \cdot \prod_{i=1}^{m}\left(-z_i\right) |$, this integral is positive infinity.
	
	\vspace{0.5em}

	\noindent \textit{Case 2}: Single pure integrator exists in $L(s)$, \textit{i.e.} $l = n-1$.
	
	When $l = n-1$, equating the constant terms in~\eqref{C_Den} yields
	\begin{equation}\label{coef1_case2}
		K  \prod_{i=1}^{m}  \left(-z_i\right) = \prod_{i=1}^{n}\left(-r_i\right).
	\end{equation}
	Equating the coefficients of terms $\tilde{s}^{n-1}$ in~\eqref{C_Den2} gives
	\begin{equation}\label{coef2_case2}
		- \prod_{i=1}^{n-1}(-p_i) + K \prod_{i=1}^{m}(-z_i)  \left( \sum_{i=1}^{m}z_i^{-1} \right) = \prod_{i=1}^{n}\left(-r_i\right) \left( \sum_{i=1}^{n} r_i^{-1} \right).
	\end{equation}
	Applying~\hyperref[lem2]{Lemma~2} to $\tilde{T}_2(\tilde{s})$ in~\eqref{C2_T2} and substituting~\eqref{coef1_case2} into the result, the CSBI  becomes
	\begin{equation}\label{Cont_Bode3}
		\begin{split}
			\dfrac{1}{2\pi}\int_{-\infty}^{\infty} |T(s)| \ \dfrac{d\omega}{\omega^2} & = 	\dfrac{1}{2\pi} \int_{-\infty}^{\infty} \ln |\tilde{T}_2\left(\tilde{s}\right)| d\tilde{\omega}\\
			& = \dfrac{1}{2} \cdot \left[ \sum_{i=1}^{m} \left| \textrm{Re} \ z_i^{-1}\right| - \sum_{i=1}^{n} \left| \textrm{Re} \ r_i^{-1} \right| \right]    + \dfrac{1}{\pi} \int_{-\infty}^{\infty} \ln \left|K \cdot 		\dfrac{\prod_{i=1}^{m}\left(-z_i\right)}{\prod_{i=1}^{n}\left(-r_i\right)}\right| d\tilde{\omega}\\
			& = \dfrac{1}{2} \cdot \left[ \sum_{i=1}^{m} \left| \textrm{Re} \ z_i^{-1}\right| - \sum_{i=1}^{n} \left| \textrm{Re} \ r_i^{-1} \right| \right].
		\end{split}
	\end{equation}	
	The first term on the RHS of~\eqref{Cont_Bode3} can be decomposed as
	\begin{equation}\label{res_1}
		\sum_{i=1}^{m} \left|\textrm{Re} \  z_i^{-1} \right| = \sum_{i} \textrm{Re} \ z^{-1}_{u_i} - \sum_{i} \textrm{Re} \ z_{s_i}^{-1},
	\end{equation}
	where $z_{u_i}$ and $z_{s_i}$ respectively denote the non-minimum phase zeros and minimum phase zeros in $L(s)$. Since the closed-loop system is stable and all the closed-loop poles $r_i$'s have negative real parts, with identities~\eqref{coef1_case2} and~\eqref{coef2_case2}, the second term on the RHS of~\eqref{Cont_Bode3} can be rewritten as
	\begin{equation}\label{eq27}
		\begin{split}
			\sum_{i=1}^{n} \left|\textrm{Re} \ r_i^{-1} \right| & = - \left[ \prod_{i=1}^{n}\left(-r_i\right) \right]^{-1} \cdot  \left[ -\prod_{i=1}^{n-1}(-p_i) + K\prod_{i=1}^{m}(-z_i) \sum_{i=1}^{m} \textrm{Re} \ z_i^{-1}\right]\\
			& = - \dfrac{ -\prod_{i=1}^{n-1}(-p_i) + K\prod_{i=1}^{m}(-z_i) \sum_{i=1}^{m} \textrm{Re} \ z_i^{-1} }{K  \prod_{i=1}^{m}  \left(-z_i\right)}\\
			& = - \sum_{i} \mathrm{Re} \ z_{s_i}^{-1} - \sum_{i} \mathrm{Re} \ z_{u_i}^{-1} + \dfrac{1}{K} \cdot \dfrac{\prod_{i=1}^{n-1}(-p_i)}{\prod_{i=1}^{m}(-z_i)}.
		\end{split}
	\end{equation}
	Combing the results in~\eqref{Cont_Bode3}, \eqref{res_1} and~\eqref{eq27}, when single integrator exists in $L(s)$, the CSBI is
	\begin{equation}\label{eq277}
		\dfrac{1}{2\pi}\int_{-\infty}^{\infty} |T(s)| \ \dfrac{d\omega}{\omega^2} =  \sum_{i} \mathrm{Re} \ z_{u_i}^{-1} - \dfrac{1}{2K} \cdot \dfrac{\prod_{i=1}^{n-1}(-p_i)}{\prod_{i=1}^{m}(-z_i)}.
	\end{equation}
	
	\vspace{0.5em}
	
	\noindent \textit{Case 3}: Two or more pure integrators exist in $L(s)$, \textit{i.e.} $0 \leq l \leq n - 2$.
	
	When $0 \leq l \leq n-2$, equating the constant terms in~\eqref{C_Den} gives the same identity as~\eqref{coef1_case2}, $K \cdot \prod_{i=1}^{m}  (-z_i) = \prod_{i=1}^{n}(-r_i)$. Equating the coefficients of terms $\tilde{s}^{n-1}$ in~\eqref{C_Den2} yields
	\begin{equation}\label{coef2_case3}
		K \prod_{i=1}^{m}(-z_i)  \left( \sum_{i=1}^{m}z_i^{-1} \right) = \prod_{i=1}^{n}\left(-r_i\right) \left( \sum_{i=1}^{n} r_i^{-1} \right).
	\end{equation}
	Since all the closed-loop poles $r_i$'s have negative real parts, with~\eqref{coef1_case2} and~\eqref{coef2_case3}, we have
	\begin{equation}\label{eq30}
			\sum_{i=1}^{n} \left|\textrm{Re} \ r_i^{-1} \right|  = - \left[ \prod_{i=1}^{n}\left(-r_i\right) \right]^{-1} \cdot  \left[ K\prod_{i=1}^{m}(-z_i) \sum_{i=1}^{m} \textrm{Re} \ z_i^{-1}\right] = - \sum_{i} \textrm{Re} \ z^{-1}_{u_i} - \sum_{i} \textrm{Re} \ z_{s_i}^{-1}.
	\end{equation}
	Applying~\hyperref[lem2]{Lemma~2} to $\tilde{T}_2(\tilde{s})$ in~\eqref{C2_T2} and substituting~\eqref{coef1_case2} and~\eqref{eq30} into the result, the CSBI for $L(s)$ with two or more integrators is
	\begin{equation}\label{Cont_Bode4}
		\begin{split}
			\dfrac{1}{2\pi}\int_{-\infty}^{\infty} |T(s)| \ \dfrac{d\omega}{\omega^2} & = 	\dfrac{1}{2\pi} \int_{-\infty}^{\infty} \ln |\tilde{T}_2\left(\tilde{s}\right)| d\tilde{\omega}\\
			& = \dfrac{1}{2} \cdot \left[ \sum_{i=1}^{m} \left| \textrm{Re} \ z_i^{-1}\right| - \sum_{i=1}^{n} \left| \textrm{Re} \ r_i^{-1} \right| \right]    + \dfrac{1}{\pi} \int_{-\infty}^{\infty} \ln \left|K \cdot 		\dfrac{\prod_{i=1}^{m}\left(-z_i\right)}{\prod_{i=1}^{n}\left(-r_i\right)}\right| d\tilde{\omega}\\
			& =\sum_{i} \textrm{Re} \ z^{-1}_{u_i}.
		\end{split}
	\end{equation}
	Summarizing the results in~\eqref{Cont_Bode2}, \eqref{eq277} and~\eqref{Cont_Bode4} leads to~\eqref{thm1_eq} in~\hyperref[thm1]{Theorem~1}. This completes the proof.
\end{proof}

	Next, we consider the scenario when the open-loop transfer function $L(s)$ is biproper, \textit{i.e.} relative degree $\nu = n-m = 0$. The CSBI of this category is related to not only the number of pure integrators $n-l$, but also the value of leading coefficient $K$.
	
	\begin{corollary}\label{cor2}
	For an open-loop transfer function $L(s)$ with relative degree $\nu = 0$ and stable closed-loop system, the continuous-time CSBI satisfies
		\begin{equation*}
			\dfrac{1}{2\pi}\int_{-\infty}^{\infty} \log |T(s)| \dfrac{d\omega}{\omega^2} =
			\begin{cases}
			\sum_{i} \mathrm{Re} \ z^{-1}_{u_i},   & \mathrm{if} \ K \neq -1 \ \mathrm{ and } \ 0 \leq l \leq n - 2; \\
			 \sum_{i} \mathrm{Re} \ z_{u_i}^{-1} - \dfrac{1}{2K} \cdot \dfrac{\prod_{i=1}^{n-1}(-p_i)}{\prod_{i=1}^{n}(-z_i)}, \hspace{16pt} & \mathrm{if} \ K \neq -1 \ \mathrm{ and }  \ l=n-1; \\
			\pm \infty, & \mathrm{otherwise}.
			\end{cases}
		\end{equation*}
	\end{corollary}

\begin{proof}[\textbf{Proof}]
	When relative degree $\nu = 0$, the open-loop transfer function $L(s)$ can be expressed as follows
	\begin{equation}\label{eq32}
		L(s) = K \cdot \dfrac{\prod_{i=1}^{n}(s - z_i)}{s^{n-l} \prod_{i=1}^{l}(s - p_i)},
	\end{equation}
	where $K \in \mathbb{R}$,  $l \leq n$, and no zero $z_i$ or pole $p_i$ is at $s = 0$. When $K = -1$, since the term $s^n$  vanishes in the denominator of complementary sensitivity function $T(s)$ in~\eqref{complementary}, applying the frequency transformation~\eqref{freq_trans} and similar coefficient manipulations as in~\eqref{C2_T2}-\eqref{eq32}, the complementary sensitivity function $T(s)$ satisfies
	\begin{equation}\label{eq33}
	T(s) = \dfrac{\prod_{i=1}^{n}(s - z_i)}{K' \cdot \prod_{i=1}^{q}(s-r_i)} = \dfrac{\prod_{i=1}^{n}(-z_i)}{K' \prod_{i=1}^{q}(-r_i)} \cdot \dfrac{\prod_{i=1}^{n}\left(\tilde{s} - z_i^{-1}\right)}{\tilde{s}^{n-q} \cdot \prod_{i=1}^{q}\left( \tilde{s} - r_i^{-1} \right)} = \tilde{T}(\tilde{s}),
	\end{equation}
	where $r_i$'s are the closed-loop poles with negative real parts, $K' \in \mathbb{R}$ is the lumped coefficient, $q < n$, and the values of $K'$ and $q$ are determined by the distributions of poles $p_i$'s and zeros $z_i$'s in $L(s)$. In general, we do not have $\prod_{i=1}^{n}(-z_i) = K' \prod_{i=1}^{q}(-r_i)$, \textit{i.e.} $T(0) = 1$; otherwise, one can derive the corresponding CSBI by applying the analysis in the proof of~\hyperref[thm1]{Theorem~1} to this specific set of $p_i$'s and $z_i$'s. Hence, applying~\hyperref[lem2]{Lemma~2} to $\tilde T(s)$ in~\eqref{eq33} yields
	\begin{equation}\label{Cont_Bode5}
	\begin{split}
	\dfrac{1}{2\pi}\int_{-\infty}^{\infty} \ln |T(s)| \ \dfrac{d\omega}{\omega^2} & = 	\dfrac{1}{2\pi} \int_{-\infty}^{\infty} \ln |\tilde{T}\left(\tilde{s}\right)| d\tilde{\omega}\\
	& = \dfrac{1}{2} \left[ \sum_{i=1}^{n} \left| \textrm{Re} \ z_i^{-1}\right| - \sum_{i=1}^{q} \left| \textrm{Re} \ r_i^{-1} \right| \right]    + \dfrac{1}{\pi} \int_{-\infty}^{\infty} \ln \left|   \dfrac{\prod_{i=1}^{n}(-z_i)}{K'  \prod_{i=1}^{q}(-r_i)}  \right| d\tilde{\omega}.
	\end{split}
	\end{equation}
	Since the last term on the RHS of~\eqref{Cont_Bode5} is unbounded, in general, the CSBI is unbounded when $K = -1$.

	When $K \neq -1$, the complementary sensitivity function $T(s)$ can be equivalently expressed as
	\begin{align}
		T_1(s) & = T(s) = \dfrac{ K  \prod_{i=1}^{n}(s - z_i) }{s^{n-l} \prod_{i=1}^{l}(s - p_i) + {K  \prod_{i=1}^{n}(s - z_i) }}, \label{v0_T1}\\	
		T_2(s) & = T(s) = \dfrac{K}{K+1} \cdot \dfrac{\prod_{i=1}^{n}(s - z_i) }{ \prod_{i=1}^n (s - r_i) }. \label{v0_T2}
	\end{align}
	Equating the denominators of~\eqref{v0_T1} and~\eqref{v0_T2} and expanding the polynomials give the following equation
	\begin{equation}\label{v0_Den1}
		\begin{split}
			s^n - \left( \sum_{i=1}^{l}p_i \right) s^{n-1} + \cdots + \prod_{i=1}^{l}\left(-p_i\right)  s^{n-l} & + K \left[s^n - \left(\sum_{i=1}^{n} z_i\right) s^{n-1} + \cdots + \prod_{i=1}^{n}\left(-z_i\right)\right] \\
			&=(K+1) \left[s^n - \left(\sum_{i=1}^{n} r_i\right)s^{n - 1} + \cdots + \prod_{i=1}^{n}\left(-r_i\right)\right].
		\end{split}
	\end{equation}
	Adopting the frequency transformation~\eqref{freq_trans} and some coefficient manipulations, \eqref{v0_T1} and~\eqref{v0_T2} can be transformed into
	\begin{align}
		\tilde{T}_1(\tilde{s}) 	& = T_1(s) = \dfrac{K \prod_{i=1}^{n} (-z_i) \cdot \prod_{i=1}^{n} (\tilde{s} - z_i^{-1})}{\prod_{i=1}^{l}(-p_i) \cdot \prod_{i=1}^{l}(\tilde{s} - p_i^{-1})+ K \prod_{i=1}^{n}(-z_i) \cdot \prod_{i=1}^{n}(\tilde{s} - z_i^{-1})}, \label{v0_T3}\\
		\tilde{T}_2(\tilde{s}) & = T_2(s) = \dfrac{K  \prod_{i=1}^{n}(-z_i)}{(K+1)  \prod_{i=1}^{n} (-r_i)  } \cdot \dfrac{\prod_{i=1}^{n} (\tilde{s} - z_i^{-1})}{\prod_{i=1}^{n}(\tilde{s} - r_i^{-1})}. \label{v0_T4}
	\end{align}
	Equating and expanding the denominators in~\eqref{v0_T3} and~\eqref{v0_T4}, we have
	\begin{equation}\label{v0_Den2}
		\begin{split}
			\prod_{i=1}^{l}(-p_i) \left[ \tilde{s}^l - \left( \sum_{i=1}^{l} p_i^{-1} \right) \tilde{s}^{l-1} + \cdots + \prod_{i=1}^{l}\left( -p_i^{-1} \right) \right] + K \prod_{i=1}^{n}(-z_i) \Bigg[ \tilde{s}^{n} - \left( \sum_{i=1}^{n} z_i^{-1} \right)\tilde{s}^{n-1} + \cdots\\
			+ \prod_{i=1}^{n}\left(-z_i^{-1}\right) \Bigg] = (K+1) \prod_{i=1}^{n}\left(-r_i\right) \left[ \tilde{s}^n - \left( \sum_{i=1}^{n} r_i^{-1} \right) \tilde{s}^{n-1} + \cdots + \prod_{i=1}^{n}\left(-r_i^{-1}\right)  \right].
		\end{split}
	\end{equation}
	Similar to the proof of~\hyperref[thm1]{Theorem~1}, we analyze the CSBIs for $L(s)$ with different amounts of pure integrators.
	
	\vspace{0.5em}	
	
	\noindent \textit{Case 1}: No pure integrator exists in $L(s)$, \textit{i.e.} $l = n$.
	
	When $l = n$, equating the constant terms in~\eqref{v0_Den1} gives, $\prod_{i=1}^{n}(-p_i) + K \prod_{i=1}^{n}(-z_i) = (K+1) \prod_{i=1}^{n}(-r_i)$.	Equating the coefficients of $\tilde{s}^{n-1}$ terms in~\eqref{v0_Den2} yields, $\prod_{i=1}^{n}(-p_i) \left( \sum_{i=1}^{n} p_i^{-1} \right) + K \prod_{i=1}^{n}(-z_i) \left( \sum_{i=1}^{n} z_i^{-1} \right) = (K+1) \prod_{i=1}^{n}(-r_i) \left( \sum_{i=1}^{n} r_i^{-1} \right)$. Applying~\hyperref[lem2]{Lemma~2} to $\tilde{T}_2(\tilde{s})$ in~\eqref{v0_T4}, the complementary sensitivity Bode integral satisfies
	\begin{equation}\label{Cont_Bode6}
			\dfrac{1}{2\pi}\int_{-\infty}^{\infty} \ln |T(s)| \ \dfrac{d\omega}{\omega^2} = \dfrac{1}{2} \left[ \sum_{i=1}^{n} \left| \textrm{Re} \ z_i^{-1}\right| - \sum_{i=1}^{n} \left| \textrm{Re} \ r_i^{-1} \right| \right]    + \dfrac{1}{\pi} \int_{-\infty}^{\infty} \ln \dfrac{\left|K \prod_{i=1}^{n}\left(-z_i\right)\right|}{\left|\prod_{i=1}^{n}\left(-p_i\right) + K \prod_{i=1}^{n}  \left(-z_i\right)\right| } d\tilde{\omega}.\\
	\end{equation}
	Bode integral in~\eqref{Cont_Bode6} is bounded when at least one $p_i = 0$, which contradicts the fact that $p_i \neq 0$ when we defined $L(s)$ in~\eqref{eq32}. Hence, the integral is unbounded when $l = n$. Further analysis on this case refers to the comments after~\eqref{Cont_Bode2}.
		
	\vspace{0.5em}	
	
	\noindent \textit{Case 2}: Single pure integrator exists in $L(s)$, \textit{i.e.} $l = n-1$.
	
	When $l = n - 1$, equating the constant terms in~\eqref{v0_Den1} gives $K \prod_{i=1}^{n} (-z_i) = (K+1) \prod_{i=1}^{n}(-r_i)$. Equating the coefficients of $\tilde{s}^{n-1}$ terms in~\eqref{v0_Den2} gives $-\prod_{i=1}^{n-1}(-p_i) + K\prod_{i=1}^{n}(-z_i)(\sum_{i=1}^{n}z_i^{-1}) = (K+1) \cdot \prod_{i=1}^{n}(-r_i)( \sum_{i=1}^{n}r_i^{-1} )$. From~\eqref{res_1} we have $\sum_{i=1}^{n} \left|\textrm{Re} \  z_i^{-1} \right| = \sum_{i} \textrm{Re} \ z^{-1}_{u_i} - \sum_{i} \textrm{Re} \ z_{s_i}^{-1}$. Since all the closed-loop poles $r_i$'s are with negative real parts, the sum $\sum_{i=1}^{n}|\textrm{Re} \ r_i^{-1}| = \prod_{i=1}^{n-1}(-p_i) / \allowbreak [K\prod_{i=1}^{n}(-z_i)] - \sum_{i=1}^{n}z_i^{-1}$. Applying~\hyperref[lem2]{Lemma~2} to $\tilde{T}_2(\tilde{s})$ in~\eqref{v0_T4}, the CSBI gives
	\begin{equation}\label{eq41}
		\dfrac{1}{2\pi}\int_{-\infty}^{\infty} \ln |T(s)| \ \dfrac{d\omega}{\omega^2} = \sum_{i} \mathrm{Re} \ z_{u_i}^{-1} - \dfrac{1}{2K} \cdot \dfrac{\prod_{i=1}^{n-1}(-p_i)}{\prod_{i=1}^{n}(-z_i)}.
	\end{equation}
	
	\vspace{0.5em}	
	
	\noindent \textit{Case 3}: Two or more pure integrators exist in $L(s)$, \textit{i.e.} $0 \leq l \leq n - 2$.
	
	When $0 \leq l \leq n - 2$, equating the constant terms in~\eqref{v0_Den1} gives $K \prod_{i=1}^{n}(-z_i) = (K+1) \prod_{i=1}^{n}(-r_i)$. Equating the coefficients of $\tilde{s}^{n-1}$ terms in~\eqref{v0_Den2} gives $K \prod_{i=1}^{n}(-z_i)  ( \sum_{i=1}^{n}z_i^{-1} ) = (K+1) \prod_{i=1}^{n}(-r_i) \cdot ( \sum_{i=1}^{n} r_i^{-1} )$. Meanwhile, we have $\sum_{i=1}^{n} \left|\textrm{Re} \  z_i^{-1} \right| = \sum_{i} \textrm{Re} \ z^{-1}_{u_i} - \sum_{i} \textrm{Re} \ z_{s_i}^{-1}$ and $\sum_{i=1}^{n} \left|\textrm{Re} \  r_i^{-1} \right| = -\sum_{i=1}^{n}z_i^{-1} =- \sum_{i} \textrm{Re} \ z^{-1}_{u_i} - \sum_{i} \textrm{Re} \ z_{s_i}^{-1}$. Hence, applying~\hyperref[lem2]{Lemma~2} to $\tilde{T}_2(\tilde{s})$ in~\eqref{v0_T4}, the complementary sensitivity Bode integral satisfies
	\begin{equation}\label{eq42}
		\dfrac{1}{2\pi}\int_{-\infty}^{\infty} \ln |T(s)| \ \dfrac{d\omega}{\omega^2} = \sum_{i} \mathrm{Re} \ z_{u_i}^{-1}.
	\end{equation}
	
	\noindent Summarizing the results in~\eqref{Cont_Bode6},~\eqref{eq41} and~\eqref{eq42} gives~\hyperref[cor2]{Corollary~2}. This completes the proof.
\end{proof}

\begin{remark}
	The results on continuous-time CSBI have been presented in~\hyperref[thm1]{Theorem~1} for systems with $\nu \geq 1$ and~\hyperref[cor2]{Corollary~2} for systems with $\nu = 0$, respectively. In general, our results, \hyperref[thm1]{Theorem~1} and \hyperref[cor2]{Corollary~2}, derived via the simplified approach match the earlier result, \hyperref[lem1]{Lemma~1}, derived by employing Cauchy integral theorem. Nevertheless, more detailed and explicit relationship between CSBI and the features of $L(s)$ and more relaxed constraints on $L(s)$ are attained by using the simplified approach. In both cases of $\nu \geq 1$ and $\nu = 0$, the values of CSBIs are mainly determined by the non-minimum phase zeros $z_{u_i}$, while the format of CSBI varies depending on the amount of pure integrators in $L(s)$. When only single pure integrator exists in $L(s)$, the value of CSBI, whose explicit expression was not reported in the previous papers, is also impacted by the leading coefficient $K$, minimum phase zeros $z_{s_i}$, as well as the poles $p_i$ in $L(s)$. Our derivations also show that the continuous-time CSBI defined in~\cite{Middleton_1991} is unbounded, when $L(s)$ does not contain any pure integrator, which is a limitation of this type of CSBI and did not receive enough attention in  recent papers~\cite{Seron_2012, Wan_CDC_2018}. Meanwhile, the analytic constraint on $\ln |T(1/s)|$, as well as the initial value constraint $T(0) \neq 0$, are not necessary in the proofs of~\hyperref[thm1]{Theorem~1} and~\hyperref[cor2]{Corollary~2}, which can also be extended to the scenario when some of the closed-loop poles $r_i$'s are on the imaginary axis.
\end{remark}

\section{Discrete-Time Complementary Sensitivity Bode Integral}\label{sec4}

Discrete-time CSBI is investigated in this section by using a simplified approach developed on the basis of~\hyperref[lem4]{Lemma~4}. Compared with the continuous-time system, since the frequency domain of discrete-time system is bounded, we do not need to worry about the unboundedness of discrete-time CSBI, and hence neither weighting function nor frequency inversion is involved in this section. When the relative degree $\nu \geq 1$ in $L(s)$, we have the following result.

\begin{theorem}\label{thm3}
	For an open-loop transfer function $L(z)$ with relative degree $\nu \geq 1$ and stable closed-loop system, the discrete-time CSBI satisfies
	\begin{equation}\label{thm3eq}
		\dfrac{1}{2\pi}\int_{-\pi}^{\pi} \log |T(z)| d\omega = \sum_{i} \log |z_{u_i}|  + \log |K|,
	\end{equation}
	where $z_{u_i}$'s are the unstable zeros in $L(z)$, and $K$ is the leading coefficient in~\eqref{dis_tf}.
\end{theorem}
\begin{proof}[\textbf{Proof}]
	When relative degree $\nu \geq 1$, the open-loop transfer function~\eqref{dis_tf} can be expressed as
	\begin{equation*}
		L(z) = K \cdot \dfrac{\prod_{i=1}^{m}(z - z_i)}{\prod_{i=1}^{n}(z - p_i)},
	\end{equation*}
	where $n \geq m + 1$. Then the discrete-time complementary sensitivity function takes the form
	\begin{equation}\label{eq44}
		T(z) = \dfrac{L(z)}{1 + L(z)} = \dfrac{K \cdot \prod_{i=1}^{m}(z - z_i)}{\prod_{i=1}^{n}(z - p_i) + K \cdot \prod_{i=1}^{m}(z- z_i)} = K \cdot \dfrac{\prod_{i=1}^{m}(z - z_i)}{\prod_{i=1}^{n}(z - r_i)}.
	\end{equation}
	Since the closed-loop system is stable and all the closed-loop poles $r_i$'s are within the unit disk, applying~\hyperref[lem4]{Lemma~4} to~\eqref{eq44}, the discrete-time CSBI satisfies	
	\begin{equation}
		\dfrac{1}{2\pi}\int_{-\pi}^{\pi} \log |T(z)| d\omega = \dfrac{1}{4\pi} \int_{-\pi}^{\pi} \log \left| K \cdot \dfrac{\prod_{i=1}^{m}(z - z_i)}{\prod_{i=1}^{n}(z - r_i)}\right|^2 d\omega = \sum_{i} \log |z_{u_i}|  + \dfrac{1}{2\pi} \int_{-\pi}^{\pi} \log | K | d\omega,
	\end{equation}
	which implies~\eqref{thm3eq} in~\hyperref[thm3]{Theorem~3}. This completes the proof.
\end{proof}

We then consider the CSBI of a biproper open-loop system, $\textit{i.e.}$ when $\nu = n - m = 0$.

\begin{corollary}\label{cor4}
	For an open-loop transfer function $L(z)$ with relative degree $\nu = 0$ and stable closed-loop system, the discrete-time CSBI satisfies
	\begin{equation}\label{cor4eq}
	\dfrac{1}{2\pi}\int_{-\pi}^{\pi} \log |T(z)| d\omega = \sum_{i} \log |z_{u_i}|  + \log \left| \dfrac{K}{1+K} \right|.
	\end{equation}
\end{corollary}
\begin{proof}
	When relative degree $\nu = 0$, the discrete-time open-loop transfer function~\eqref{dis_tf} can be expressed as follows	
	\begin{equation*}
		L(z) = K \cdot \dfrac{\prod_{i=1}^{n}(z - z_i)}{\prod_{i=1}^{n}(z - p_i)},
	\end{equation*}
	where $K = \lim_{z\rightarrow \infty}L(z)$. The complementary sensitivity function then becomes
	\begin{equation}\label{eq47}
		T(z) = \dfrac{L(z)}{1 + L(z)} = \dfrac{K \cdot \prod_{i=1}^{n} (z - z_i) }{\prod_{i=1}^{n}(z - p_i) + K \cdot \prod_{i=1}^{n}(z - z_i)} = \dfrac{K}{1 + K} \cdot \dfrac{\prod_{i=1}^{n}(z - z_i)}{\prod_{i=1}^{n}(z - r_i)}.
	\end{equation}
	When $K = -1$, the order of denominator in~\eqref{eq47} will be less than $n$, \textit{i.e.} at least one closed-loop pole $r_i$ is out of unit circle and at infinity, which implies that the closed-loop system is not causal. Hence, in practice, the leading coefficient $K = -1$ is not allowed when $\nu = 0$~\cite{Wu_TAC_1992}, though one can still obtain a bounded value by applying~\hyperref[lem4]{Lemma~4} to~\eqref{eq47} or computing the integral~\eqref{eq7} directly. When $K \neq -1$, since the closed-loop system is stable and all $r_i$'s are inside the unit disk, applying~\hyperref[lem4]{Lemma~4} to~\eqref{eq47}, the discrete-time CSBI satisfies
	\begin{equation}
		\dfrac{1}{2\pi}\int_{-\pi}^{\pi} \log |T(z)| d\omega = \dfrac{1}{4\pi} \int_{-\pi}^{\pi} \ln \left| \dfrac{K}{1 + K} \cdot \dfrac{\prod_{i=1}^{n}(z - z_i)}{\prod_{i=1}^{n}(z - r_i)} \right|^2 d\omega =  \sum_{i} \log |z_{u_i}|  + \dfrac{1}{2\pi} \int_{-\pi}^{\pi} \log \left| \dfrac{K}{1+K} \right| d\omega,
	\end{equation}
	which implies~\eqref{cor4eq} in~\hyperref[cor4]{Corollary~4}. This completes the proof.
\end{proof}
\begin{remark}
	The results on discrete-time CSBI have been presented in~\hyperref[thm3]{Theorem~3} for systems with $\nu \geq 1$ and in~\hyperref[cor4]{Corollary~4} for systems with $\nu = 0$, respectively. These results, derived by using the simplified approach, match the previous results in~\cite{Sung_IJC_1989} generally. For both cases, $\nu \geq 1$ and $\nu = 0$, the CSBI is proportional to the sum of the logarithms of unstable or non-minimum phase zeros. However, the difference between the second terms on the RHS of~\eqref{thm3eq} and~\eqref{cor4eq} was not noted in the previous papers.
\end{remark}

\section{Illustrative Examples}\label{sec5}

Illustrative examples that examine the previous theorems and corollaries are given in this section. First, we consider the following open-loop transfer function with two pure integrators

\begin{equation}\label{eq50}
	L_1(s) = -1.164 \times 10^{-4} \cdot \dfrac{  (s-10)(s + 0.0625)}{s^2 \cdot (s+ 10)},
\end{equation}
where the relative degree $\nu = 1$, and a non-minimum phase zero is located at $s = 10$. The closed-loop complementary sensitivity function is
\begin{equation}\label{eq51}
	T_1(s) = \dfrac{-1.164 \times 10^{-4} \cdot (s-10)(s+0.0625)}{(s+10)(s^2 + 1.149\times 10^{-4} \cdot s + 7.725 \times 10^{-6})},
\end{equation}
which is a closed-loop stable plant with three closed-loop poles located at $s = 5.745 \times 10^{-5} \pm 2.697 \times 10^{-3} i$ and $-10$. The numerical integration of $T_1(s)$ in~\eqref{eq51} gives $(2\pi)^{-1}\int_{-\infty}^{\infty} \ln |T_1(s)| / \omega^2 {d\omega} \break \approx 0.1000$. Applying~\hyperref[thm1]{Theorem~1} to~\eqref{eq50} yields $(2\pi)^{-1} \int_{-\infty}^{\infty} \ln |T_1(s)| / \omega^2 d\omega = \sum_{i} \mathrm{Re} \ z_{u_i}^{-1} = {1}/{10} = 0.1$, which matches the numerical result.

Next, we consider an open-loop transfer function with only one pure integrator, \textit{i.e.} $l = n-1$,
\begin{equation}\label{eq52}
	L_2(s) = \dfrac{-5.77 (s-10)(s+1)}{s(s+10)(s+1)},
\end{equation}
which, with relative degree $\nu = 1$, has a non-minimum phase zero at $s = 10$. The closed-loop complementary sensitivity function of~\eqref{eq52} is
\begin{equation}\label{eq53}
	T_2(s) = \dfrac{-5.77 (s-10)}{s^2 + 4.23s + 57.7},
\end{equation}
which is closed-loop stable with two closed-loop poles at $s = -2.115 \pm 7.296 i$. The numerical integration of~\eqref{eq53} gives $(2\pi)^{-1} \int_{-\infty}^{\infty} \ln \left|T_2(s) \right| / \omega^2 d\omega \approx 0.0133$. Applying~\hyperref[thm1]{Theorem~1} to~\eqref{eq52} yields $(2\pi)^{-1} \int_{-\infty}^{\infty} \ln \left|T_2(s) \right| / \omega^2 d\omega = \sum_{i} \mathrm{Re} \ z_{u_i}^{-1} - (2K)^{-1} \cdot \prod_{i=1}^{n-1}(-p_i) / \prod_{i=1}^{m}(-z_i) = 77 / 5770 \approx 0.0133$, which also matches the numerical result.

Meanwhile, by following  similar procedures as above, one can easily verify that the CSBI of the open-loop transfer function $L_3(s) = -2.0348 \cdot(s-1) / (s^2 + 3s + 2)$ is unbounded, which validates the last case in~\hyperref[thm1]{Theorem~1}. In the end, an illustrative example is given to examine~\hyperref[cor4]{Corollary~4}, which was rarely noted before. Consider a biproper discrete-time system $L_4(z) = 2(z+2) / (z+0.5)$. By using numerical integration, the CSBI of $L_4(z)$ gives $(2\pi)^{-1} \int_{-\pi}^{\pi} \log |L_4(z) / (1+L_4(z))| d\omega \approx 0.4150 \approx \log 2 + \log(2 / 3)$, which justifies~\hyperref[cor4]{Corollary~4}. For brevity, more examples that verify Corollary 2 and Theorem 3 are omitted in this note and left to the interested readers.

\section{Conclusions}\label{sec6}

A simplified approach for analyzing complementary sensitivity trade-offs in both continuous-time and discrete-time systems was proposed in this note. A comprehensive relationship between CSBIs and the features of open-loop transfer functions was interpreted by using this simplified approach. A few illustrative examples were presented to justify the results.

\section*{Acknowledgment}
This work was partially supported by AFOSR and NSF. The authors would specially acknowledge the readers and staff on arXiv.org.

\bibliographystyle{IEEEtran}
\bibliography{ref}

\end{document}